\documentclass[pra, aps, floatfix, twocolumn, preprintnumbers, superscriptaddress]{revtex4}
\usepackage{hyperref}
\usepackage{verbatim}
\usepackage{amsmath}
\usepackage{latexsym}
\usepackage{revsymb}
\usepackage{yfonts}

\usepackage{natbib}
\usepackage{amsfonts}
\usepackage{amsmath}
\usepackage{amssymb}
\usepackage{amsthm}
\usepackage{graphicx}
\usepackage{bm}
\usepackage{bbm}

\newcommand{\be}{\begin{eqnarray} \begin{aligned}}
\newcommand{\ee}{\end{aligned} \end{eqnarray} }
\newcommand{\benn}{\begin{eqnarray*} \begin{aligned}}
\newcommand{\eenn}{\end{aligned} \end{eqnarray*} }

\newcommand{\bc}{\begin{center}}
\newcommand{\ec}{\end{center}}

\newcommand{\id}{\mathbb{I}}

\newcommand{\tr}{\mathop{\mathrm{tr}}\nolimits}
\newcommand{\Tr}{\mathop{\mathrm{Tr}}\nolimits}

\newtheorem{theorem}{Theorem}[section]

\newtheorem{lemma}[theorem]{Lemma}
\newtheorem{definition}[theorem]{Definition}

\newtheorem{corollary}[theorem]{Corollary}

\newcommand{\cM}{\ensuremath{\mathcal{M}}}
\newcommand{\cY}{\ensuremath{\mathcal{Y}}}
\newcommand{\cP}{\ensuremath{\mathcal{P}}}
\newcommand{\me}{\mathbf{f}}

\newcommand{\argmax}{\mathop{\mathrm{argmax}}\nolimits}

\usepackage{amsfonts}

\def\Complex{\mathbb{C}}

\def\id{\mathbb{I}}

\def\01{\{0,1\}}

\newcommand{\eps}{\varepsilon}
\newcommand{\ket}[1]{|#1\rangle}
\newcommand{\bra}[1]{\langle#1|}

\newcommand{\proj}[1]{|#1\rangle\langle#1|}

\newcommand{\inp}[2]{\langle{#1}|{#2}\rangle} 

\newcommand{\ignore}[1]{}

\newenvironment{sdp}[2]{
\smallskip
\begin{center}
\begin{tabular}{ll}
#1 & #2\\
subject to
}
{
\end{tabular}
\end{center}
\smallskip
}

\newcommand{\enc}{\ket{\Psi_{y_0y_1}}}
\newcommand{\senc}{\proj{\Psi_{y_0y_1}}}
\newcommand{\rot}{X_d^{y_0} Z_d^{y_1}}

\newcommand{\xd}{X_d}
\newcommand{\zd}{Z_d}

\newcommand{\hmin}{{\ensuremath{{\rm H}}_{\infty}}}
\newcommand{\hmineps}{{\ensuremath{{\rm H}}_{\infty}^\eps}}

\begin{document}

\title{ {Does ignorance of the whole imply ignorance of the parts?}\\ {\normalsize --- Large violations of non-contextuality in quantum theory}}
\author{Thomas Vidick}
\affiliation{Computer Science division, UC Berkeley, USA}
\email{vidick@eecs.berkeley.edu}
\author{Stephanie Wehner}
\affiliation{Center for Quantum Technologies, National University of Singapore, 2 Science Drive 3, 117543 Singapore}
\email{wehner@nus.edu.sg}
\date{\today}

\begin{abstract}
	A central question in our understanding of the physical world is how our knowledge of the whole relates to our knowledge of the individual parts. One aspect of this question is the following: to what extent does ignorance about a whole preclude knowledge of at least one of its parts?
	Relying purely on classical intuition, one would certainly be inclined to conjecture that a strong ignorance of the whole cannot come without significant 
	ignorance of at least one of its parts. Indeed, we show that this reasoning holds in any  non-contextual hidden variable model (NC-HV).
	Curiously, however, such a conjecture is \emph{false} in quantum theory: we provide an explicit example where a large ignorance about the whole can coexist with an almost perfect knowledge of each of its parts. More specifically, we provide a simple information-theoretic inequality satisfied in any NC-HV, but which can be \emph{arbitrarily} violated by quantum 
	mechanics. 
	Our inequality has interesting implications for quantum cryptography. 
\end{abstract}
\maketitle

In this note we examine the following seemingly innocent question: does one's ignorance about the whole
necessarily imply ignorance about at least one of its parts?
Given just a moments thought, the initial reaction is generally to give a positive answer. 
Surely, if one cannot know the whole, then one should be able to point to an unknown part.
Classically, and more generally for any deterministic non-contextual hidden variable model, our intuition turns out to be correct: ignorance about the whole does indeed imply the existence of a specific part which is unknown, so that one can point to the source of one's ignorance. 
However, we will show that in a quantum world this intuition is flawed. 

\section*{THE PROBLEM}
Let us first explain our problem more formally. Consider two dits $y_0$ and $y_1 \in \{0,\ldots,d-1\}$, where the string $y = y_0y_1$ plays the role of the whole, and $y_0$, $y_1$ are the individual parts. Let $\rho_y$ denote an encoding of the string $y$ into a classical or quantum state. In quantum theory,
$\rho_y$ is simply a density operator, and in a NC-HV model it is a preparation $\cP_y$ described by a probability distribution over hidden variables $\lambda \in \Lambda$. Let $P_Y$ be a probability distribution over $\{0,\ldots,d-1\}^{2}$, and imagine that with probability $P_Y(y)$ we are given the state $\rho_y$. 
The optimum probability of guessing $y$ given its encoding $\rho_y$, which lies in a register $E$, can be written as 
\begin{align}\label{eq:guessingProb}
P_{\rm guess}(Y|E) = \max_{\{\cM\}} \sum_{y \in \{0,\ldots,d-1\}^{ 2}} P_{Y}(y)\, p(y|\cM,\cP_y)\ ,
\end{align}
where 
$p(y|\cM,\cP_y)$ is the probability of obtaining outcome $y$ when measuring the preparation $\cP_y$ with $\cM$,
and the maximization is taken over all $d^2$-outcome measurements allowed in the theory. 
In the case of quantum theory, for example, 
the maximization is taken over POVMs $\cM = \{M_y\}_y$ and $p(y|\cM,\cP_y) = \tr(M_y \rho_y)$.
The guessing probability is directly related to the conditional min-entropy 
$\hmin(Y|E)$ through the equation~\cite{krs:entropy}
\begin{align}
\hmin(Y|E) := - \log P_{\rm guess}(Y|E)\ .
\end{align}
This measure plays an important role in quantum cryptography and is the relevant measure of information
in the single shot setting corresponding to our everyday experience, as opposed to the asymptotic setting captured by the von Neumann entropy. 
A closely related variant is the smooth min-entropy $\hmineps(Y|E)$ which can be thought of as being like $\hmin(Y|E)$ except with some
small error probability $\eps$. 
The main question we are interested in can then be loosely phrased as:

\begin{quote}
How does $\hmin(Y=Y_0Y_1|E)$ (ignorance about the whole)
relate to $\hmin(Y_C|EC)$, for $C \in \{0,1\}$ (ignorance about the parts)?
\end{quote}

Here the introduction of the additional random variable $C$ is crucial, and it can be understood as a pointer to the part of $Y$ about which there is large ignorance (given a large ignorance of the whole string $Y$); see Figure~\ref{fig:game} for an illustration of this role. 
It is important to note that the choice of $C$ should be consistent with the encoding prior to its definition. That is, whereas $C$ 
may of course depend on $Y_0,Y_1$ and the encoding $E$, the reduced state on registers holding $Y_0,Y_1$ and $E$ after tracing out $C$ should remain the same.
In particular, this condition states that $C$ cannot be the result of a measurement causing disturbance to the encoding register; if we were allowed to destroy information
in the encoding we would effectively alter the original situation.

\section*{RESULTS}

{\bf An inequality valid in any NC-HV model.}

We first show that classically, or more generally in any non-contextual hidden variable model~\cite{endnote23},
ignorance about the whole really \emph{does} imply ignorance about a part.
More specifically, we show that for any random variable $Y=Y_0Y_1$ and side information $E$, there exists a random variable $C \in \{0,1\}$ such
that
\begin{align}\label{eq:classicalSplitting}
\hmin(Y_C|EC) \gtrsim\frac{\hmin(Y_0Y_1|E)}{2}\ .
\end{align}
This inequality can be understood as an information-theoretic analogue of Bell inequalities to the question of non-contextuality.
Classically, this inequality is known as the \emph{min-entropy splitting inequality}, and 
plays an important role in the proof of security of some (classical) cryptographic primitives~\cite{juerg:splitting, serge:new}.
The proof of~\eqref{eq:classicalSplitting} is a straightforward extension to the case of standard 
NC-HV models~\cite{klyachko:nc,andreas:nc} of a 
classical technique known as min-entropy splitting first introduced by Wullschleger~\cite{juerg:splitting}, and we defer details to the appendix. 

The fact that $C$ is a random variable, rather than being deterministically chosen, is important, and an example will help clarify its role. Consider $Y$ uniformly distributed over $\{0,\ldots,d-1\}^{2}$ and $E = Y_0$ with probability $1/2$, and $Y_1$ with probability $1/2$. In this case it is easy to see that \emph{both} $Y_0$ and $Y_1$ can be guessed from $E$ with average success probability $1/2+1/(2d)$, so that $\hmin(Y_0|E) = \hmin(Y_1|E) \approx 1$, which is much less than $\hmin(Y|E)\approx \log d$. However, define $C$ as $0$ if $E=Y_1$ and $1$ if $E=Y_0$. Then it is clear that $\hmin(Y_C|EC) = \log d$, as we are always asked to predict the variable about which we have no side information at all! In this case the random variable $C$ ``points to the unknown'' by being correlated with the side information $E$, but is entirely consistent with our knowledge about the world: by tracing out $C$ we recover the initial joint distribution on $(Y,E)$. This also highlights the important difference between the task we are considering and the well-studied random access codes~\cite{nayak:rac,nayak:original}, in which the requirement is to be able to predict one of $Y_0,Y_1$ (adversarially chosen) from their encoding; for this task it has been demonstrated that there is virtually no asymptotic difference between classical and quantum encodings (see below for a discussion). 

It is interesting to note that~\eqref{eq:classicalSplitting} still holds if we consider a somewhat ``helpful'' physical model 
in which in addition to the encoding one might learn a small number of ``leaked'' bits of information about $Y$. More specifically, if the NC-HV discloses $m$ extra bits of information 
then it follows from the chain rule for the min-entropy (see appendix) that
\begin{align}
\hmin(Y_C|EC) \gtrsim\frac{\hmin(Y_0Y_1|E)}{2} - m\ .
\end{align}

{\bf Violation in quantum theory.} Our main result shows that~\eqref{eq:classicalSplitting} is violated in the strongest possible sense by quantum theory. More specifically, we provide an explicit construction that demonstrates this violation:
Let $Y = Y_0Y_1$ be uniformly distributed over $\{0,\ldots,d-1\}^{2}$. Given $y=y_0y_1\in\{0,\ldots,d-1\}^2$, define its encoding $\rho_{y_0y_1}^E = \proj{\Psi_y}$ as

\begin{align}\label{eq:encoding}
\ket{\Psi_y} := X^{y_0}_d Z^{y_1}_d \ket{\Psi}\ ,
\end{align}
where $X_d$ and $Z_d$ are the generalized Pauli matrices and 
\begin{align}
	\ket{\Psi} := 
	\frac{1}{\sqrt{2\left(1+\frac{1}{\sqrt{d}}\right)}}
	(\ket{0} + F\ket{0})\ ,
\end{align}
with $F$ being the matrix of the Fourier transform over $\mathbb{Z}_d$. Since we are only interested in showing a quantum violation, 
we will for simplicity always assume that $d$ is prime~\cite{endnoteAdd}. The system $YE$ is then described by the ccq-state
\begin{align}
\rho_{Y_0Y_1E} = \frac{1}{d^2} \sum_{y_0,y_1} \proj{y_0} \otimes \proj{y_1} \otimes \rho_{y_0y_1}^E\ .
\end{align}
We first prove that $\hmin(Y|E) = \log d$ for our choice of encoding.
We then show the striking fact that, even though the encoding we defined gives very little information about the whole string $Y$, for any adversarially chosen random variable $C$ (possibly correlated with our encoding) one can guess $Y_C$ from its encoding $\rho_E$ with essentially 
constant probability.
More precisely, for any ccqc-state $\rho_{Y_0Y_1EC}$, with $C \in \{0,1\}$, that satisfies the consistency relation $\tr_C(\rho_{Y_0Y_1EC}) = \rho_{Y_0Y_1E}$, we have
\begin{align}
\hmin(Y_C|EC) \approx 1
\end{align}
for \emph{any} sufficiently large $d$. 
This shows that the inequality~\eqref{eq:classicalSplitting} can be violated arbitrarily (with $d$), giving a striking example of the malleability
of quantum information. 
What's more, it is not hard to show that this effect still holds even for $\hmineps$, for constant error $\eps$,
and a ``helpful'' physical model leaking $m \approx c \log d$ bits of information with $c < 1/2$. 
Hence, the violation of 
the inequality~\eqref{eq:classicalSplitting} has the appealing feature of being very robust. Indeed, for any number of bits $m$ a NC-HV might leak in
addition, we could find a $d$ to ensure a violation.

\begin{center}
\begin{figure}
	\includegraphics[scale=0.5]{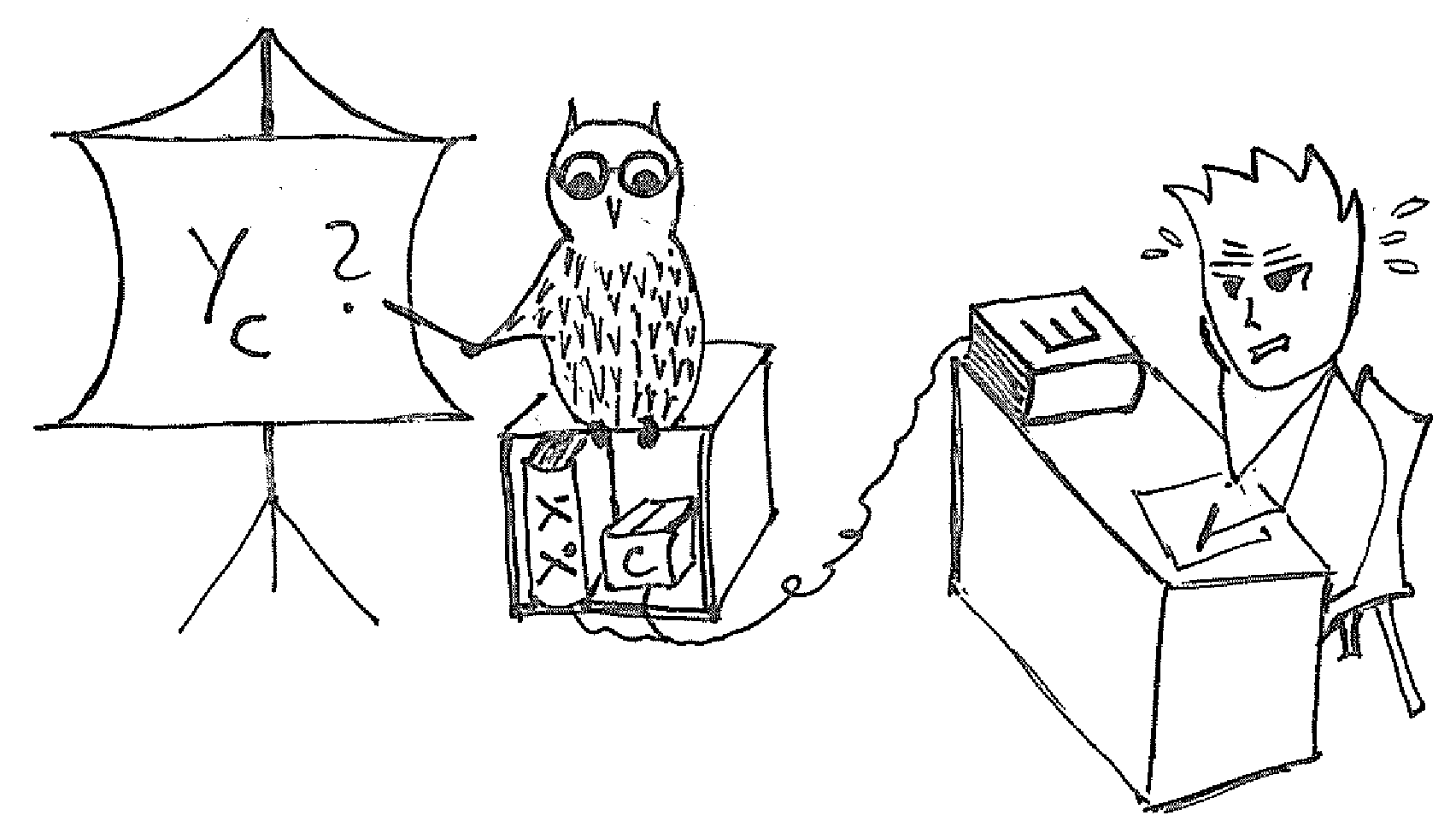}
	\caption{{\scriptsize Intuitively, one can also understand our result in terms of a game between Bob and a malicious challenger, the Owl. Imagine Bob is taking a philosophy class teaching him knowledge about $Y$, clearly chosen uniformly at random. Unfortunately, he never actually attended and had insufficient time to prepare for his exam. Luckily, however, he has been given some encoding $E$ of the possible answers $Y_0Y_1$, hastily prepared by his old friend Alice. When entering the room, he had to submit $E$ for inspection to the challenger who knows $Y_0$, $Y_1$ as well as the encoding Alice might use. After inspection, the challenger may secretly keep a 
	system $C$, possibly correlated with $E$, but such that the reduced system on $Y_0$, $Y_1$ and $E$ looks untampered with.
	It is immediately obvious to the challenger that Bob must be ignorant about the whole of $Y_0Y_1$. But can it always measure and point to a $C=c$ 
	such that Bob is ignorant about $Y_C$? That is, can it always detect Bob's ignorance by challenging him to output a single $Y_C$?
	Classically, this is indeed possible: ignorance about the whole of $Y_0Y_1$ implies significant ignorance about one of the parts, $Y_C$.
	However, a quantum Bob could beat the Owl.}}
	\label{fig:game}
\end{figure}
\end{center}

\section*{PROOF OF THE QUANTUM VIOLATION}

We now provide an outline of the proof that the encoding specified in~\eqref{eq:encoding} leads to a quantum violation of the splitting
inequality~\eqref{eq:classicalSplitting}; for completeness, we provide a more detailed derivation in the appendix.
Our proof proceeds in three steps: first, by computing $\hmin(Y|E)$ we show that the encoding does indeed not reveal much information about the whole.
Second, we compute the optimal measurements for extracting $Y_0$ and $Y_1$ on average, and show that these measurements perform equally well
for any other prior distribution on $Y$.
Finally, we show that even introducing an additional system $C$ does not change one's ability to extract $Y_C$ from the encoding.

\smallskip
\noindent
{\bf Step 1: } 
Very intuitively, ignorance about the whole string already follows from Holevo's theorem and the fact
that we are trying to encode 2 dits into a $d$-dimensional quantum system.
To see this more explicitly, recall that $\hmin(Y|E) = \log d$ is equivalent to showing that
$P_{\rm guess}(Y|E) = 1/d$. 
From~\eqref{eq:guessingProb} we have that this guessing probability is given by the solution to the following semidefinite program (SDP)
\begin{sdp}{maximize}{$\frac{1}{d^2} \sum_{y_0,y_1} \tr\left(M_{y_0y_1}\senc\right)$}
& $M_{y_0y_1} \geq 0 \mbox{ for all } y_0,y_1$\ ,\\
& $\sum_{y_0,y_1} M_{y_0,y_1} = \id$\ .
\end{sdp}
The dual SDP is easily found to be
\begin{sdp}{minimize}{$\Tr(Q)$}
& $Q \geq \frac{1}{d^2} \senc \mbox{ for all } y_0,y_1$\ .
\end{sdp}
Let $v_{\rm primal}$ and $v_{\rm dual}$ be the optimal values of the primal and dual respectively. By the property of weak duality, $v_{\rm dual} \geq v_{\rm primal}$ always holds. Hence, to prove our result, we only need to find a primal and dual solutions for which $v_{\rm primal} = v_{\rm dual} = 1/d$.  
It is easy to check that $\hat{Q} = \id/d^2$ is a dual
solution with value $v_{\rm dual} = \tr(\hat{Q}) = 1/d$.
Similarly, consider the measurement $M_{y_0y_1} = \senc/d$. Using Schur's lemma, one can directly verify that
$\sum_{y_0,y_1} M_{y_0y_1} = \id$, giving $v_{\rm primal} = 1/d$. The claimed value of the conditional min-entropy follows.

\smallskip
\noindent
{\bf Step 2: }
A similar argument, exploiting the symmetries in the encoding, can be used to show that
\begin{align}\label{eq:racRecover}
	P_{\rm guess}(Y_0|E) = P_{\rm guess}(Y_1|E) = \frac{1}{2} + \frac{1}{2\sqrt{d}}\ . 
\end{align}
The measurements that attain these values are given by the eigenbases of $Z_d$ and $X_d$ respectively.

As a remark to quantum information theorists, note that this means that our encoding doubles
as a random access encoding of the string $y$ into a $d$-dimensional quantum state $\rho_y$ with
probability~\eqref{eq:racRecover} to recover $y_0$ or $y_1$.
For $d=2$, such encodings have previously been considered in the realm of contextuality
as a reinterpretation of the CHSH inequality~\cite{chsh,Spekkens2009}.
However, we note that this is \emph{not} what is surprising here, as
there exists an obvious classical random access encoding for 2 dits into a single dit (see discussion on $C$ above), with recovery probability
$1/2 + 1/(2d)$.

Simply computing~\eqref{eq:racRecover} is hence insufficient for our purposes. Let us write $\{\ket{y_0},\,y_0 \in \{0,\ldots,d-1\}\}$ for
the eigenbasis of $Z_d$, and note that its Fourier transform $\{F\ket{y_1},\,y_1 \in \{0,\ldots,d-1\}\}$ is then the eigenbasis of $X_d$. 
Exploiting the symmetries in our problem, it is straightforward to verify that for all $y_0,y_1 \in \{0,\ldots,d-1\}$
\begin{align}
	|\langle y_0|\Psi_{y_0y_1}\rangle|^2 = 
	|\langle y_1|F^\dagger|\Psi_{y_0y_1}\rangle|^2 = 
	\frac{1}{2} + \frac{1}{2 \sqrt{d}}\ .
\end{align}
An important consequence of this is that for \emph{any} other prior distribution $P_{y_0y_1}$,
measurement in the $Z_d$ eigenbasis
distinguishes the states
\begin{align}
	\sigma_{y_0} = \sum_{y_1} P_{y_0y_1}(y_0,y_1) \senc\ ,
\end{align}
with probability at least $1/2 + 1/(2\sqrt{d})$, even when the distribution is unknown.
A similar argument can be made for the marginal states $\sigma_{y_1}$ and measurement in the $X_d$ eigenbasis.

\smallskip
\noindent
{\bf Step 3: }
It now remains to show that, for any possible choice of an additional classical system 
$C$~\cite{endnote25},
one can still guess $Y_C$ from the encoding with a good success probability: one cannot construct a $C$ which would ``point to the unknown''.
Note that we may express the joint state with any other system $C$ as
\begin{align}
\rho_{Y_0Y_1 E C} = \frac{1}{d^2} \sum_{y_0 y_1} \proj{y_0} \otimes \proj{y_1} \otimes \rho_{y_0y_1c}^{EC}\ ,
\end{align}
for some states $\rho_{y_0y_1c}^{EC}$ on registers $E$ and $C$.
Since the reduced state on $Y_0$,$Y_1$ and $E$ should be the same for any $C$
we have by the fact that $Y_0$ and $Y_1$ are classical that $tr_{C}(\rho_{y_0y_1c}^{EC}) = \senc$. Since $\senc$ is a pure state, this implies
that $\rho_{y_0y_1c}^{EC} = \senc \otimes \sigma_{y_0 y_1}^C$. 
Now imagine that we were to perform some arbitrary measurement on $C$, whose outcome would supposedly point to an unknown substring. But this merely creates a different distribution $P_{y_0y_1}$ 
over encoded strings, and we already know from the above that we can still succeed in retrieving either $y_0$ or $y_1$ 
with probability at least $1/2 + 1/(2\sqrt{d})$ by making a measurement in the $X_d$ or $Z_d$ basis respectively.
Hence for large $d$ we have a recovery probability of roughly $1/2$, implying
\begin{align}
	\hmin(Y_0|EC=0) \approx \hmin(Y_1|EC=1) \approx 1\ ,
\end{align}
which is our main claim.

Note that the consistency condition, which states that our choice of $C$ should be compatible
 with the original situation and should not affect the reduced state, is important, and
makes our task non-trivial. As an example, consider our construction for $d=2$. In that case 
the encoding states lie in the XZ-plane of the Bloch sphere. 
Imagine now that we measured the encoding register $E$ in the eigenbasis of $\sigma_y$, and let the outcome be $C$. But for 
any measurement in the eigenbasis of $\sigma_y$ we observe entirely random outcomes, and the post-measurement states trivially no longer carry any information about the 
encoded string. Indeed, any choice of $C$ would do if we are allowed to destroy information in such a manner.

\section*{IMPLICATIONS FOR CRYPTOGRAPHY}

Our result answers an interesting open question in quantum cryptography~\cite{chris:talk}, 
namely whether min-entropy splitting can still be performed
when conditioned on quantum instead of classical knowledge. This technique was used to deal with classical side information $E$
in~\cite{serge:bounded,serge:new}.
Our example shows that quantum min-entropy splitting is impossible, even when we would
be willing to accept subtracting a large error term on the r.h.s. of~\eqref{eq:classicalSplitting}. 
This tells us that classical protocols that rely on such statements may become insecure in the presence of quantum side information, 
and highlights the importance of so-called min-entropy sampling results of~\cite{kr:sampling}
used in quantum cryptography~\cite{KoeWehWul09} instead. It also indicates that contextuality may play a more important role in our understanding
of the possibilities and limits of quantum cryptography than previously thought.

\section*{DISCUSSION}

The first indication that something may be amiss when looking at knowledge from a quantum perspective was given by Schr{\"o}dinger~\cite{schroedinger:eprGerman}, who
pointed out that one can have \emph{knowledge} (not ignorance) about the whole, while still being ignorant about the parts~\cite{endnote26}.
Here, we tackled this problem from 
a very different direction, starting with the premise that one has ignorance about the whole. 

Our results show that contextuality is responsible for much more significant effects than have previously been noted. 
In particular, it leads to arbitrarily large quantum violations of~\eqref{eq:classicalSplitting}, which can be understood
as a Bell-type inequality for non-contextuality. This is still true even for a somewhat ``helpful'' physical model, leaking 
additional bits of information. To our knowledge, this is the first \emph{information-theoretic} inequality distinguishing
NC-HV models from quantum theory. Our question and perspective are completely novel, and we hope that our observations will lead to an increased understanding
of the role of contextuality. In this work, we have considered standard NC-HVs in which all HVs can be decomposed as convex combinations of extremal HVs which give deterministic outcomes for effects (see appendix). It is an interesting open question whether our results can be generalized to very general models that distinguish between measurement and preparation contextuality~\cite{Spekkens2005}.

%

At the heart of our result lies the fact that contextuality allows for strong forms of complementarity in quantum mechanics (often conflated with 
uncertainty~\cite{js:urvsnl}),
which intuitively is responsible for allowing the violation of~\eqref{eq:classicalSplitting}. 
Typically, complementarity is discussed by considering examples of properties of a physical system that one may be able to determine individually, but which cannot all be learned at once. In spirit, this is similar to the notion of a random access encoding where we could determine either property 
$Y_0$ or $Y_1$ quite well, but not all of $Y$. However, as discussed above this can also be true classically, in a probabilistic sense.
We would thus like to emphasize the novelty of our perspective, as we approach the problem from the other end, and first demonstrate
the general result that in an NC-HV ignorance about the whole always implies ignorance about a part. We then show that
in a quantum world, this principle is violated in the strongest possible sense, even with respect to an additional system $C$. 
One could think of this as a much more robust way of capturing the intuitive notion of complementarity~\cite{js:inprep}.

Finally, it is an interesting open question whether our inequality can be experimentally verified. Note that this made difficult by the fact that our aim would be to test \emph{ignorance} rather than knowledge. However, it is conceivable that such an experiment can be performed by building a larger cryptographic protocol whose security relies on being ignorant about one of the parts of a string $Y$ created during that protocol~\cite{endnote27}.
A quantum violation could then be observed by breaking the security of the protocol, and exhibiting \emph{knowledge} (rather than ignorance) about some information that could not have been obtained if the protocol was secure.

\acknowledgments
We thank Jonathan Oppenheim, Christian Schaffner, Tony Short, Robert Spekkens and CQT's ''non-local club'' for useful comments. 
We are particularly grateful to Tony Short for pointing out that our problem could more easily 
be explained by means of the game depicted in Figure~\ref{fig:game}.
TV was supported by ARO Grant W911NF-09-1-0440 and NSF Grant CCF-0905626.
SW was supported by the National Research Foundation, and the Ministry of Education, Singapore.
TV is grateful to CQT, Singapore, for hosting him while part of this work was done. SW is grateful for an invitation from the Mittag-Leffler Institute, Sweden, where part of this work was performed.

\appendix
\bigskip

\noindent
In this appendix, we provide a detailed derivation of our results. To this end, we first provide some more detailed background on the entropic quantities we use in Section~\ref{sec:entropy}. 
In Section~\ref{sec:lhv} we show that the splitting inequality~\eqref{eq:splitting} is satisfied in any deterministic non-contextual hidden variable model (NC-HV model for short). This is a minor twist
on the existing classical proof~\cite{serge:new} due to Wullschleger~\cite{juerg:splitting}.
Finally, in Section~\ref{sec:encoding}, we proceed to prove our main result, that there exists
a quantum encoding which strongly violates the splitting inequality~\eqref{eq:splitting}.

\section{Entropy measures}\label{sec:entropy}

Throughout, we will measure information in terms of the min-entropy, which is directly related to the 
\emph{guessing probability} $P_{\rm guess}(Y|E)$~\cite{krs:entropy}, where $Y$ is a classical string ranging in the set $\cY$ and $E$ an auxiliary system. It is defined as the maximum probability with which one can predict the whole string $Y$, given the system $E$. The maximization is over all possible observations, or measurements, on $E$; these vary depending on the physical model (e.g. classical or quantum) under consideration. 

\begin{definition}\label{def:ncguess}
Let $Y$ be a classical random variable with distribution $P_Y$ taking values in a set $\cY$, and $\{ \cP_{y}\}_{y\in \cY}$ any set of preparations on $E$. Then the maximum guessing probability of $Y$ given $E$ is defined as
\begin{align}\label{eq:pguessdef}
	P_{\rm guess}(Y|E) = \max_{\{\cM_y\}} \sum_{y\in \cY} P_Y(y)\,p(y|\cP_y,\cM_y)\ ,
\end{align}
where the maximum is taken over all measurements $\cM = \{\cM_{y}\}$ allowed in the model.
\end{definition}
For instance, in the case of quantum mechanics we simply have 
\begin{align}\label{eq:quantumguess}
P_{\rm guess}(Y|E) := \max_{\substack{\{M_y\}_{y \in \mathcal{Y}}\\
\forall y M_y \geq 0\\
\sum_y M_y = \id}}  \sum_y P_Y(y) \tr\left(M_y \rho^E_y\right)\ ,
\end{align}
where the maximization is taken over all POVMs~\cite{endnote28} and
$\rho_y^E$ denotes the reduced state of the system on $E$, when $Y=y$. For \emph{classical} side-information $E$ this expression simplifies to
\begin{align}\label{eq:classicalguess}
P_{\rm guess}(Y|E) := \mathbb{E}_{e\leftarrow E}\left[\max_{y} P_{Y|E=e}(y)\right]\ .
\end{align}
In other words, for classical side information, the optimal guessing measurement is to 
simply output the $y$ which is most likely given the classical value $e$.

For the case of classical and quantum theories it is known that the guessing probability directly relates to the conditional 
min-entropy~\cite{krs:entropy}.
Here, we follow the operational approach of~\cite{ts:entropy}, and define the conditional min-entropy for an arbitrary theory with classical $Y$ as
\begin{align}\label{eq:pguesshmin}
\hmin(Y|E) := - \log P_{\rm guess}(Y|E)\ .
\end{align}
For the case of quantum systems, the conditional min-entropy was first introduced by Renner~\cite{renato:diss} as a way to measure randomness conditioned on an adversary's knowledge. 
The min-entropy can also be defined when $Y$ is quantum itself~\cite{renato:diss}, but we will not need it here.
In the quantum setting, we will also use a smoothed version of the quantum conditional min-entropy, defined for any $\eps>0$ as:
\begin{align}
\hmineps(Y|E) = \max_{\tilde{\rho}_{YE} \in \mathcal{B}_\eps(\rho_{YE})} \hmin(Y|E)_{\tilde{\rho}_{YE}}\ ,
\end{align}
where the maximization is taken over all (subnormalized) states $\tilde{\rho}_{YE}$ within $\eps$ trace distance of $\rho_{YE}$.
A similar definition could be made for
arbitrary theories using the distance defined in~\cite{ts:entropy}, but we will not require it here.

The conditional min-entropy has a number of appealing properties, which for any NC-HV model essentially follow from its operational interpretation,
and also hold in the quantum setting~\cite{renato:diss}. First of all consider the min-entropy of classical $YZ$ conditioned on side information $E$. Clearly, since guessing $Y$ \emph{and} $Z$ can only be more difficult then guessing $Y$ alone, we have $P_{\rm guess}(YZ|E) \leq P_{\rm guess}(Y|E)$. Translated, this gives monotonicity of the min-entropy
\begin{align}\label{eq:monotone}
\hmin(YZ|E) \geq \hmin(Y|E)\ .
\end{align}
Similarly, to guess $Y$ \emph{and} $Z$ from $E$ one strategy would be to guess $Z$ (in the worst case, choosing $Z=z$ with $z \in \mathcal{Z}$ taken from the uniform distribution) and then try to guess $Y$ knowing $Z$. In terms of guessing probabilities, this means that 
$P_{\rm guess}(YZ|E) \geq P_{\rm guess}(Y|EZ)/|\mathcal{Z}|$. Translated, we obtain the chain rule
\begin{align}\label{eq:chain}
\hmin(Y|EZ) \geq \hmin(YZ|E) - \log |\mathcal{Z}|\ .
\end{align}
A final property that will be important to us is that, as a direct consequence of~\eqref{eq:pguesshmin} we may also write the min-entropy as
\begin{align}\label{eq:classicalMeasure}
\hmin(Y|E) = \min_{\mathcal{M}} \hmin(Y|\mathcal{M}(E))\ ,
\end{align}
where the minimization is taken over all measurements $\mathcal{M}$, and $\hmin(Y|\mathcal{M}(E))$ is the min-entropy
conditioned on the classical information obtained by measuring $E$ with $\mathcal{M}$.

\section{A splitting inequality valid in any NC-HV model}\label{sec:lhv}

Before turning to the proof of the generalized splitting inequality, let us briefly review what is meant by a non-contextual model. 
In any physical theory, we can imagine that a system is prepared according to some \emph{preparation} $\cP$, on which we later make \emph{measurements} $\cM$. 
Each such measurement can be viewed as a collection of elementary effects $\me$. The exact form of the effects depends on the model one considers.
For example, in quantum theory the effects are simply given by POVM elements. 
A particularly useful effect is given by the so-called \emph{unit effect} $\id$, corresponding to the identity
in the quantum or classical setting. Hence to any effect one can associate a two-outcome measurement $\mathcal{M}_{\me} = \{\me,\id - \me\}$. 

When discussing non-contextuality, this
 measurement is typically interpreted as a question one might pose to the underlying physical system and has two answers, ``yes'' for $\me$ and ``no'' for $\id - \me$.  
We hence also refer to $\me$ as a question.
Of course, one might consider measurements that ask many questions simultaneously, that is, they consist of many individual effects.
Two effects are called \emph{compatible} if the corresponding questions can be answered simultaneously without causing disturbance to the underlying physical system, in the sense that we would obtain
the same answers again were we to ask the same questions repeatedly. 

A set of mutually compatible effects/questions is thereby called a \emph{context}. For example, if $\me_1$ is compatible with $\me_2$ the set $\mathcal{C}_1 = \{\me_1,\me_2\}$ is called 
a \emph{context}. Similarly, if $\me_1$ is compatible with $\me_3$, then the set $\mathcal{C}_2 = \{\me_1,\me_3\}$ is also a context. Note, however, that in such a scenario it can still 
be that $\me_2$ and $\me_3$ are \emph{not} compatible. That is, any effect can be part of multiple \emph{distinct} contexts. 

For each effect in a particular context, one can pose the question $\me$ by making the measurement $\mathcal{M}_{\me}$ defined above. Informally, a model is 
called \emph{non-contextual} if the answer to question $\me_1$ will always be the same in both contexts, whether $\mathcal{M}_{\me_2}$ \emph{or} $\mathcal{M}_{\me_3}$ are performed simultaneously (which is possible by definition of being compatible). In our example this means that if 
 were we to make measurement $\mathcal{M}_{\me_1}$ in context $\mathcal{C}_1$, or context $\mathcal{C}_2$, we would always obtain the same distribution on outcomes.

\subsection{Classical theory}\label{sec:purelyClassical}

Recall that the 
phenomenon of min-entropy splitting guarantees that, if a string $Y_0Y_1$ has high min-entropy then there is a way to split it by introducing a binary random variable $C$ such that the string $Y_{C}$ 
has about half as much min-entropy as $Y_0Y_1$. Classically, min-entropy splitting follows from the following statement.

\begin{lemma}[\cite{serge:new}, Lemma~4.2]\label{lem:classicalsplitting} Let $\eps>0$ and $Y_0,Y_1$ two random variables such that $\hmineps(Y_0Y_1|E)\geq \alpha$, where $E$ is classical. 
	Then, there exists a binary random variable $C$ such that $\hmineps(Y_{C}C|E) \geq \alpha/2$.
\end{lemma}
Using the chain rule~\eqref{eq:chain} and the monotonicity~\eqref{eq:monotone} of the min-entropy one immediately obtains 
the statement of min-entropy splitting
\begin{align}\label{eq:splittingStatement}
\hmineps P(Y_{C}|EC)\geq \alpha/2 -1 - \log 1/\eps'
\end{align}

\subsection{Non-contextual hidden variable models}\label{sec:nchv}

Typically, in a non-contextual hidden variable model it is assumed that a preparation $\cP$ is simply a distribution over hidden variables $\Lambda$, and a measurement then corresponds to ``reading out'' such hidden variables. Each outcome event $k \in K$ is associated with a corresponding effect $\me_k$, where intuitively $\me_k$ ``reads out'' the hidden variables by mapping a certain subset of possible hidden variables to the outcome $k \in K$.
In contrast, some works consider more generalized scenarios known as \emph{ontological models}~\cite{Spekkens2005}. The main difference here is that these hidden variable models can locally model even contextual theories, but specify explicit conditions to make these generalized theories non-contextual again.

 In this section we show that the splitting inequality holds in any standard \emph{deterministic} NC-HV, which is the definition taken in most previous work, as in e.g.~\cite{klyachko:nc,andreas:nc}. We will, however, phrase our result in the general language of non-contextual models as introduced in~\cite{Spekkens2005}, restricting our attention to those models which are deterministic.

\subsubsection{Background}

Very intuitively, a non-contextual \emph{ontological model} for an operational theory associates intrinsic attributes to every physical system, which are supposed to exist independently of the particular context in which the system might be observed. These attributes are described by a set of hidden variables~\cite{endnote29}
$\lambda \in \Lambda$. Hence for us a hidden variable model consists of the following:
\begin{enumerate}
\item A set of hidden variables $\Lambda$. 
\item For every preparation $\cP$ in the physical theory, 
a probability distribution $p(\lambda|\cP)$ over $\lambda\in\Lambda$.
\item For every $\ell$-outcome measurement $\cM$, and hidden variable $\lambda \in \Lambda$, a probability distribution $p(k|\lambda,\cM)$ over $k \in [\ell] := \{1,\ldots,\ell\}$. 
\end{enumerate}
The model is indeed a model for the physical theory if it accurately predicts the outcome distribution of any measurement on any preparation, i.e. performing measurement $\cM$ on preparation $\cP$ produces outcome $k$ with probability 
\begin{align}
p(k|\cP,\cM) = \sum_\lambda p(k|\lambda,\cM )p(\lambda|\cP)\ ,
\end{align}
where for notational simplicity we assume that $\Lambda$ is discrete.


\emph{Effects.} We adopt the common notion that measurements are a collection of elementary effects. Here, an 
effect is a linear functional $\me_k: \Lambda \rightarrow [0,1]$, mapping hidden variables to outcomes. 
As is common in the study of non-contextuality~\cite{andreas:nc}, we will consider only measurements which are a collection
of deterministic effects $\me_k: \Lambda \rightarrow \{0,1\}$. 
That is, we effectively work with a \emph{deterministic} model. Much more general scenarios are certainly possible~\cite{Spekkens2005}
but we will not consider them here. 
Note that a deterministic model does not mean that there is no more randomness:
preparations are given as probability distributions over hidden variables and hence we generally do observe non-deterministic outcomes
when measuring a preparation.
Of particular importance is the unit effect $\id$ (i.e, the identity), which obeys $\id(\lambda) = 1$ 
for all $\lambda \in \Lambda$. A measurement is thus a collection $\mathcal{M} := \{\me_k\mid \sum_k \me_k = \id\}$, where we usually index the effects by the outcome that they give in $\mathcal{M}$. 
We write the probability of obtaining the outcome $k$ using
measurement $\mathcal{M}$ containing the effect $\me_k$ as
\begin{align}
p(k|\lambda,\me_k) := p(k|\lambda,\mathcal{M}) = \me_k(\lambda)\ .
\end{align}
Note that with every effect, we can again associate a two-outcome measurement $\mathcal{M}_{\me} = \{\me,\id - \me\}$\, where without loss of generality we label $\me$ using the outcome '1' and $\id - \me$ using the outcome '0'. When concerned with such a measurement $\mathcal{M}_{\me}$ we thus also 
use $p(1|\lambda,\me)$ and $p(0|\lambda,\id - \me)$ to denote the probabilities of obtaining outcomes '1' and '0' respectively.

\medskip

\emph{Extensions.} Often we wish to relate one physical system to another. For example, we may wish to perform an additional independent experiment such as flipping a coin.
Given a system with a set of hidden variables $\Lambda$, we allow its extension to a second system in the following way: if $\Lambda'$ is another set of hidden variables used to describe another physical system, 
then the combined system will have hidden variables $\Lambda \times \Lambda'$. For every preparation $\cP$ on the original system, we say that $\cP'$ is an extension of $\cP$ in the combined system if for every $\lambda\in \Lambda$ 
\begin{align}
	p(\lambda|\cP) = \sum_{\lambda'\in \Lambda'} p'((\lambda,\lambda')|\cP')\ . 
\end{align}
A measurement $\cM'$ is similarly said to extend $\cM$ as long as 
\begin{align}
p(j|\lambda,\cM) = \sum_{\lambda'} p'(j|(\lambda,\lambda'),\cM')\ .
\end{align}

\emph{Preparations.} To study our problem, we will assume that there is an implicit prior distribution on preparations $\cP$ describing prior knowledge about the state of the system under consideration. 
More specifically, we will be concerned with encodings of a string $y$ into preparations $\cP_{y}$, where the probability $P_Y(y)$ of choosing the string $y$ translates into a prior probability on the preparation as
\begin{align}
	p(\cP_y) := P_{Y}(y)\ .
\end{align}


\subsubsection{Splitting inequality}

We are now ready to generalize Lemma~\ref{lem:classicalsplitting} to any deterministic NC-HV model. The analogue of~\eqref{eq:splittingStatement} is then an easy corollary. Note that in this statement, the conditional min-entropy is understood as being defined through the guessing probability~\eqref{eq:pguessdef} as in equation~\eqref{eq:pguesshmin}. This assumes given a fixed distribution $P_{Y_0Y_1}$ on the strings $y_0y_1$, through which a prior distribution on the preparations $\cP_{y_0y_1}$ follows as explained at the end of Section~\ref{sec:nchv}.

\begin{theorem}\label{thm:mainsplitting}
Let a NC-HV model $\mathfrak{M}$ be given, with corresponding set of hidden variables $\Lambda$. Let $Y=Y_0Y_1$ be two classical random variables each taking values in a finite set $\cY$, and $\{ \cP_{y_0y_1}\}_{(y_0,y_1)\in\cY^2}$ a corresponding fixed set of preparations on a register $E$ such that
\begin{align}
	\hmin(Y_0Y_1|E) \geq \alpha\ .
\end{align}
Then there exists an extended model $\mathfrak{M}'$ over the set of hidden variables $\Lambda'= \Lambda \times \{0_C, 1_C\}$, and a set of preparations $\cP'_{y_0y_1c}$, for $c\in\{0,1\}$, extending the $\cP_{y_0y_1}$ and such that
\begin{align}
\hmin(Y_{C}\, C|E) \geq \frac{\alpha}{2}\ .\label{eq:splitting}
\end{align}
\end{theorem}

\begin{proof} 
Recall that we assume a prior distribution on the preparations given by $p(\cP_{y_0 y_1}) := P_{Y_0Y_1}(y_0y_1)$. This lets us define the guessing probability, which by assumption is such that
\begin{align}\label{eq:wholeBound}
2^{-\alpha} \geq P_{\rm guess}(Y|E)\ .
\end{align}
To rewrite the r.h.s. in terms of hidden variables, first of all note that given the prior distribution over preparations
we can write the probability of a particular hidden variable $\lambda \in \Lambda$ as
\begin{align}\label{eq:lambdaProb}
p(\lambda) &= \sum_{y_0y_1} p(\cP_{y_0y_1}) p(\lambda|\cP_{y_0y_1})\ .
\end{align}
Fix a measurement $\cM = \{\me_k\}_k$, where we indexed the effects by their outcome in the measurement. By definition, the probability of observing the outcome $k$ when $\cM$ is performed on the preparation $\cP_{y_0y_1}$ is
\begin{align}\label{eq:pkGy}
p(k|\cP_{y_0y_1},\cM) = \sum_{\lambda} p(\lambda|\cP_{y_0y_1}) p(k|\lambda,\me_k)\ ,
\end{align}
The overall probability of observing the outcome $k$ when $\cM$ is performed on the preparation $\cP$ corresponding to the mixture of the preparations $\cP_{y_0y_1}$ with associated probabilities $p(\cP_{y_0y_1})$ is then
\begin{align}
p(k) := p(k|\cP) =  \sum_{y_0y_1} p(\cP_{y_0y_1}) p(k|\cP_{y_0y_1})\ .
\end{align}
Note that by definition the hidden variables $\lambda$ give deterministic outcomes under the measurement of any effect, and hence
$	p(k|\lambda,\me_k) = \lambda_{\me_k}$,
where $\lambda_{\me_k}$ is $1$ if the measurement $\{\me_k,\id-\me_k\}$ deterministically produces the outcome '$\me_k$' when performed on a system in state $\lambda$, and $0$ otherwise. 
Using Bayes' rule twice we obtain
\begin{align}
&p(k) p(\cP_{y_0y_1}|k) = p(\cP_{y_0y_1}) p(k|\cP_{y_0y_1})\notag\\
&\qquad=p(\cP_{y_0y_1}) \sum_{\lambda} p(\lambda|\cP_{y_0y_1}) p(k|\lambda,\me_k)\notag\\
&\qquad= \sum_{\lambda} p(\lambda)
p(k|\lambda,\me_k) p(\cP_{y_0y_1}|\lambda)\ .
\end{align}
Using~\eqref{eq:classicalguess} and~\eqref{eq:classicalMeasure} we obtain that for any measurement $\mathcal{M} = \{\me_k\}$, the guessing probability of $Y_0Y_1$ is determined by the maximum posterior probability of any string $y_0y_1$, conditioned on obtaining the outcome $k$ when measuring $\cP$ with $\cM$, so that~\eqref{eq:wholeBound} implies 
\begin{align}
2^{-\alpha} &\geq \sum_k p(k) \max_{y_0y_1} p(\cP_{y_0y_1}|k)\notag\\
 & = \sum_\lambda p(\lambda) \max_{y_0y_1} \left[\sum_k p(k|\lambda,\me_k) p(\cP_{y_0y_1}|\lambda)\right]\ .\label{eq:maxbound}
\end{align}
where in order to invert the summations over $k$ and $\lambda$ with the maximization we used the fact that for any $k$, there exists exactly one $\lambda$ such that $\lambda_{\me_k}=1$ and vice-versa, so that the summation over $\lambda$ (resp. over $k$) which is after the $\max$ in the expressions above contains exactly one term. This is a consequence of the fact that the $\me_k$ form a measurement, so that $\sum_k \me_k = 1$, together with the variables $\lambda$ being deterministic, so that $p(k|\lambda,\me_k)$ can only be either $0$ or $1$.

We now need to define the additional single-bit random variable $C$, which is intuitively supposed to designate which of the two halves, $y_0$ or $y_1$, the preparation $\cP_{y_0y_1}$ contains the least amount of information about, so that we can indeed lower-bound the min-entropy $\hmin(Y_C C|\cP)$. For this we allow $C$ to be correlated with the preparation
$\cP_{y_0y_1}$. 

In order to accommodate $C$, we extend the set of hidden variables $\Lambda'$ as $\Lambda' \times \{0_C,1_C\}$. 
Define $q_1 = \sum_\lambda p(\lambda|\cP_{y_0y_1})$, where the sum ranges over all $\lambda$ such that $\sum_{y_0}p(\cP_{y_0y_1}|\lambda) \geq 2^{-\alpha/2}$, and $q_0=1-q_1$. Note that $q_0$ can be computed by the same summation, but now ranging over all $\lambda$ such that $\sum_{y_0}p(\cP_{y_0y_1}|\lambda) < 2^{-\alpha/2}$. Define two preparations as follows: 
\begin{itemize}
\item $\cP_{y_0y_11}$ is defined through the distribution 
	\begin{align}
&p((\lambda,1_C)|\cP_{y_0y_11}) = \\
&\bigg\{\begin{array}{cl} p(\lambda|\cP_{y_0y_1})/q_1 &\text{if }\sum_{y_0}p(\cP_{y_0y_1}|\lambda) \geq 2^{-\alpha/2}\\[2mm] 0&\text{otherwise}\end{array}\nonumber
\end{align}
 and $p((\lambda,0_C)|\cP_{y_0y_11}) = 0$ for every $\lambda$.
\item $\cP_{y_0y_10}$ is defined analogously by
	\begin{align}
 &p((\lambda,0_C)|\cP_{y_0y_10}) = \\
 &\bigg\{\begin{array}{cl} p(\lambda|\cP_{y_0y_1})/q_0& \text{if }\sum_{y_0}p(\cP_{y_0y_1}|\lambda) < 2^{-\alpha/2}\\[2mm] 0&\text{otherwise}\end{array}\nonumber
 \end{align}
 and $p((\lambda,1_C)|\cP_{y_0y_10}) = 0$ for every $\lambda$.
\end{itemize}
Finally, we define the preparation $\cP_{y_0y_1\,C}$ as the mixture of $\cP_{y_0y_11}$ with probability $q_1$, and of $\cP_{y_0y_10}$ with probability $q_0$. Note that the preparation $\cP_{y_0y_1C}$ is indeed an extension of $\cP_{y_0y_1}$ in the new theory, as $p((\lambda,0_C)|\cP_{y_0y_1C})+p((\lambda,1_C)|\cP_{y_0y_1C}) = p(\lambda|\cP_{y_0y_1})$. Finally, we update the prior on preparations by setting $p(\cP_{y_0y_1 1}) = q_1\, p(\cP_{y_0y_1})$ and $p(\cP_{y_0y_1 0}) = q_0 \,p(\cP_{y_0y_1})$,
so that 
\begin{align}
p(\cP_{y_0y_1}) = p(\cP_{y_0y_10})+p(\cP_{y_0y_1 1})\ .
\end{align}
One can check that with these definitions, whenever  $\sum_{y_0}p(\cP_{y_0y_1}|\lambda) \geq 2^{-\alpha/2}$ we have, using Bayes' rule twice,
\begin{align}
 &p\big(\cP_{y_0 y_1 1} | (\lambda,1_C)\big)\nonumber\\
 & = \frac{p((\lambda,1_C)|\cP_{y_0y_1 1})\, p(\cP_{y_0y_1 1})}{p((\lambda,1_C))}\nonumber\\
 &= \frac{(p(\lambda|\cP_{y_0y_1})/q_1)\cdot  (q_1 p(\cP_{y_0y_1}))}{p(\lambda)}\nonumber \\
 &= p(\cP_{y_0y_1} | \lambda)\label{eq:bayes2}
 \end{align}
and $0$ otherwise, where for the second equality we used $p(\lambda) = p((\lambda,1_C))$ for all those $\lambda$ such that $p((\lambda,1_C)|\cP_{y_0y_1 1})$ is not zero. 

\medskip

From this point on, our proof follows very closely the classical proof of Lemma~\ref{lem:classicalsplitting}. By definition, for every $y_1$ and every $\lambda$, we have that
\begin{align}
\sum_{y_0} p(\cP_{y_0 y_1 0}|(\lambda,0_C)) < 2^{-\alpha/2}\label{eq:c0}
\end{align}
It does not seem possible to similarly bound $\sum_{y_1} p(\cP_{y_0 y_11}|(\lambda,1_C))$, but it is not necessary either, as we do not have access to this quantity directly. Rather, let $\cM=\{\me_k\}_k$ be any $2d$-outcome measurement; as in~\eqref{eq:maxbound} we need to bound
\begin{align}
\sum_\lambda p(\lambda) \max_{y_0 } \left[\sum_k p(k|(\lambda,1_C),\me_k) \sum_{y_1} p(\cP_{y_0y_1 1}|(\lambda,1_C))\right]\nonumber
\label{eq:maxbound2}
\end{align}
Note that by definition, $\sum_{y_0} p(\cP_{y_0y_11}|(\lambda,1_C))$ is either $0$ or at least $2^{-\alpha/2}$, so that for all $y_0$, $y_1$ and $\lambda$, we have the trivial bound
\begin{align}
&p(\cP_{y_0y_11}| (\lambda,1_C)) \nonumber\\
&\leq  \max_{y_0y_1} p(\cP_{y_0y_11}|(\lambda,1_C))\, 2^{\alpha/2}  \sum_{y_0} p(\cP_{y_0y_11}|(\lambda,1_C))\nonumber\\
&= \max_{y_0y_1}  p(\cP_{y_0y_1}|\lambda)\, 2^{\alpha/2}  \sum_{y_0} p(\cP_{y_0y_11}|(\lambda,1_C))
\end{align}
where for the last equality we used~\eqref{eq:bayes2}. Summing this equation over all $y_1$ and combining it with~\eqref{eq:maxbound} lets us bound~\eqref{eq:maxbound2} by $2^{-\alpha} 2^{\alpha/2} \cdot 1 = 2^{-\alpha/2}$. 
 This bound together with~\eqref{eq:c0} proves the theorem.
\end{proof}

The fact that min-entropy splitting holds in any NC-HV model now follows as a corollary from Theorem~\ref{thm:mainsplitting} and the fact that the chain rule~\eqref{eq:chain} and monotonicity~\eqref{eq:monotone} of the min-entropy also hold
for NC-HV models. 

\begin{corollary}\label{cor:mainsplitting}
Let a NC-HV model $\mathfrak{M}$ be given, with corresponding set of hidden variables $\Lambda$. Let $Y=Y_0Y_1$ be two classical random variables each taking values in a finite set $\cY$, and $\{ \cP_{y_0y_1}\}_{(y_0,y_1)\in\cY^2}$ a corresponding fixed set of preparations on a register $E$. Then there exists an extended model $\mathfrak{M}'$ over the set of hidden variables $\Lambda'= \Lambda \times \{0_C, 1_C\}$, and a set of preparations $\cP_{y_0y_1c}$, for $c\in\{0,1\}$, extending the $\cP_{y_0y_1}$ such that
\begin{align}\label{eq:NCsplittingStatement}
\hmin(Y_{C}|EC)\geq 	\frac{\hmin(Y_0Y_1|E)}{2} -1\ .
\end{align}
\end{corollary}

To see that this equality is robust is now again an immediate consequence of the chain rule~\eqref{eq:chain} and monotonicity property~\eqref{eq:monotone}, which tell us that when we obtain some additional classical information $A=a$ with $a \in \mathcal{A}$ we have
\begin{align}
	\hmin(Y_{C}|EAC)&\geq 
	\hmin(Y_{C}A|EC) - \log |\mathcal{A}|\\
	&\geq \hmin(Y_{C}|EC) - \log |\mathcal{A}|\\
	&\geq \hmin(Y_0Y_1|E)/2 - \log|\mathcal{A}| - 1\ .
\end{align}
That is, a secretly helpful NC-HV leaking a small number $m = \log |\mathcal{A}|$ bits of additional information does not decrease the min-entropy by more than $\log|\mathcal{A}|$ bits.


\section{Splitting is violated by quantum mechanics}\label{sec:encoding}

We are now ready to show that the splitting inequality~\eqref{eq:splitting} is violated by quantum mechanics in a very strong sense.
To this end, we first construct a particular quantum encoding of two dits into one qudit. 

\subsection{The encoding}

Consider the encoding $E: \{0,\ldots,d-1\}^{\times 2}\rightarrow \Complex^d$ given by
\begin{align}\label{eq:encodingApp}
E(y_0,y_1) = \enc := \rot\ket{\Psi}\ ,
\end{align}
where $\xd$ and $\zd$ are the generalized Pauli matrices given by their actions on an orthonormal basis $\{\ket{y_0}, y_0 \in \{0,\ldots,d-1\}\}$
\begin{align}
	\xd\ket{y_0} &= \ket{y_0 + 1 \mod d}\ ,\\
	\zd\ket{y_0} &= \omega^{y_0} \ket{y_0}\ ,
\end{align}
with $\omega = \exp(2 \pi i/d)$, and
\begin{align}
\ket{\Psi} := \frac{1}{\sqrt{2\left(1 + \frac{1}{\sqrt{d}}\right)}} \left(\ket{0} + F\ket{0}\right)\ ,
\end{align}
with $F$ denoting the Quantum Fourier transform operator over $\mathbb{Z}_d$. Note that $\xd = F \zd F^\dagger$. We also refer to the eigenbasis
of $\zd$ as the \emph{computational basis} and the eigenbasis of $\xd$ as the \emph{Fourier basis}. Below, it will be convenient to note that
$\zd$ acts as the cyclic shift operator in the eigenbasis of $\xd$, and vice versa. Throughout, we will assume that $d$ is prime.

Imagine a source that chooses $y_0,y_1 \in \cY := \{0,\ldots,d-1\}$ uniformly 
at random
and emits $\enc$, corresponding to the ccq-state
\begin{align}\label{eq:source}
&\rho_{Y_0Y_1E}\nonumber\\
&:= \frac{1}{d^2} \sum_{y_0,y_1} 
\underbrace{\proj{y_0}}_{Y_0} \otimes \underbrace{\proj{y_1}}_{Y_1} \otimes \underbrace{\senc}_{E}\ .
\end{align}
Throughout, we will consider the probability that we guess $Y_0Y_1$ or the individual entries $Y_0$ and $Y_1$
given the register $E$.
We begin by showing that for our specific encoding the probability of guessing \emph{both} entries $Y_0,Y_1$ is
small.
\begin{lemma}
For the ccq-state $\rho_{Y_0Y_1E}$ given by~\eqref{eq:source}
\begin{align}
P_{\rm guess}(Y_0Y_1|E) = \frac{1}{d}\ .
\end{align}
\end{lemma}
\begin{proof}
Computing the probability of guessing both bits is equivalent to solving the semidefinite program (SDP)
\begin{sdp}{maximize}{$\frac{1}{d^2} \sum_{y_0,y_1} \tr\left(M_{y_0y_1}\senc\right)$}
& $M_{y_0y_1} \geq 0 \mbox{ for all } y_0,y_1$\ ,\\
& $\sum_{y_0,y_1} M_{y_0,y_1} = \id$\ .
\end{sdp}
The dual SDP is easily found to be
\begin{sdp}{minimize}{$\Tr(Q)$}
& $Q \geq \frac{1}{d^2} \senc \mbox{ for all } y_0,y_1$\ .
\end{sdp}
Let $v_{\rm primal}$ and $v_{\rm dual}$ be the optimal values of the primal and dual respectively. Note that by 
weak duality we have $v_{\rm dual} \geq v_{\rm primal}$. Since $\senc$ is a pure state, $\hat{Q} = \id/d^2$ is a feasible
dual solution with value $\tr(\hat{Q}) = 1/d$. 

We now show that $\hat{Q}$ is in fact optimal, by constructing a solution to the primal that achieves the same value.
Let $M_{y_0y_1} = \senc/d$. Clearly, $M_{y_0y_1} \geq 0$ for all $y_0$ and $y_1$, and by Schur's lemma we have
\begin{align}
\sum_{y_0,y_1} M_{y_0y_1} &= \frac{1}{d} \sum_{y_0,y_1} \rot \proj{\Psi} (\rot)^\dagger\\
&= \id\ .
\end{align}
Hence, our choice of operators is a feasible primal solution with primal value $1/d$ which concludes our claim.
\end{proof}

We now show that the probability of retrieving any of the individual entries $Y_0$ and $Y_1$ is nevertheless
quite large. To this end, let us first establish the following simple lemma.

\begin{lemma}\label{lem:indProb}
For the encoding defined in~\eqref{eq:encodingApp} we have for all $y_0, y_1 \in \{0,\ldots,d-1\}$
\begin{align}
|\inp{y_0}{\Psi_{y_0y_1}}|^2 &= 
\frac{1}{2} + \frac{1}{2 \sqrt{d}}\ , \\
|\bra{y_1}F^\dagger \ket{\Psi_{y_0y_1}}|^2 &= 
\frac{1}{2} + \frac{1}{2 \sqrt{d}}\ .
\end{align}
\end{lemma}
\begin{proof}
First of all, note that for all $y_0$ and $y_1$
\begin{align}
\langle y_0\enc &= \bra{0}(\xd^{y_0})^\dagger \enc\\
&=\bra{0}\zd^{y_1}\ket{\Psi} = \inp{0}{\Psi}\ ,
\end{align}
where we have used the fact that $\zd\ket{0} = \ket{0}$. 
Similarly, we have
\begin{align}\
\bra{y_1}F^\dagger \enc &= \omega^{- y_1 a} \bra{y_1} F^\dagger \zd^{y_1}\ket{\Psi} \\
&=
\omega^{- y_1 a} 
\bra{0}F^\dagger (\zd^{y_1})^\dagger \zd^{y_1}\ket{\Psi}\\
&=
\omega^{- y_1 a} 
\bra{0}F^\dagger\ket{\Psi}\\
&=
\omega^{- y_1 a} 
\inp{0}{\Psi}\ ,
\end{align}
with $\omega = \exp(2 \pi i/d)$, where the first equality follows from the fact that $F\ket{y_1}$ is an 
eigenvector of $\xd$, and the last equality by noting that $F\ket{\Psi} = \ket{\Psi}$. 
It thus remains to compute 
\begin{align}
\inp{0}{\Psi} &= \frac{1}{\sqrt{2 \left(1 + \frac{1}{\sqrt{d}}\right)}} \left(\inp{0}{0} + \bra{0}F\ket{0}\right)\\
&= \frac{1}{\sqrt{2}} \sqrt{1 + \frac{1}{\sqrt{d}}}\ ,
\end{align}
from which our claim follows.
\end{proof}

It is now straightforward to compute the maximum probabilities that we retrieve $Y_0$ and $Y_1$ correctly.
\begin{lemma}
For the cq-states $\rho_{Y_0E}$ and $\rho_{Y_1E}$ given by the reduced states of~\eqref{eq:source}
we have
\begin{align}
P_{\rm guess}(Y_0|E) = P_{\rm guess}(Y_1|E) = \frac{1}{2} + \frac{1}{2 \sqrt{d}}\ .
\end{align}
\end{lemma}
\begin{proof}
We first show our claim for $P_{\rm guess}(Y_0|E)$.
Consider the state corresponding to an encoding of $y_0$ given by
$\sigma_{y_0} := \frac{1}{d} \sum_{y_1} \senc$.
As before, we can express the winning probability as an SDP with primal
\begin{sdp}{maximize}{$\frac{1}{d} \sum_{y_0} \tr\left(M_{y_0}\sigma_{y_0}\right)$}
& $M_{y_0} \geq 0 \mbox{ for all } y_0$\ ,\\
& $\sum_{y_0} M_{y_0} = \id$\ .
\end{sdp}
We now show that without loss of generality, the optimal measurement has an extremely simple form.
First of all note that $[\sigma_{y_0},\zd^a] = 0$ for all $a$ and $y_0$ since 
\begin{align}
\sigma_{y_0} = \frac{1}{d} \sum_{y_1} \zd^{y_1} \proj{\Psi_{y_00}} (\zd^{y_1})^\dagger\ .
\end{align}
Hence, if $\{M_{y_0}\}_{y_0}$ is an optimal solution then so is the measurement given
by $\hat{M}_{y_0} = \frac{1}{d} \sum_{a} \zd^a M_{y_0} (\zd^a)^\dagger$.
Thus without loss of generality we may assume that the optimal measurement operators are diagonal in the 
computational basis.
Now consider the largest term corresponding to $\hat{M}_{{\rm max}}$ and $\sigma_{\rm max}$ such that
\begin{align}
\tr\left(\hat{M}_{\rm max} \sigma_{\rm max}\right) \geq \tr\left(\hat{M}_{y_0} \sigma_{y_0}\right)
\end{align}
for all $y_0$. Since all measurement operators are Hermitian, we can expand
$\hat{M}_{\rm max} = \sum_{j}\lambda_j \proj{j}$
in its eigenbasis. We may now in turn consider the element $\proj{j}$ which has the largest overlap with
$\sigma_{\rm max}$.  That is, choose 
\begin{align}\label{eq:maxM}
	m = \argmax_j \bra{j}\sigma_{\rm max}\ket{j}\ ,
\end{align}
that is, $\bra{m}\sigma_{\rm max}\ket{m} \geq \bra{j}\sigma_{\rm max}\ket{j}$
for all $j$. 
Clearly, we have that
\begin{align}\label{eq:upperBoundVal}
P_{\rm guess}(Y_0|E) \leq \bra{m}\sigma_{\rm max}\ket{m}\ .
\end{align}
It remains to prove that this inequality is tight. Without loss of generality assume that $\sigma_{\rm max} = \sigma_0$, 
any other case will follow by a simple relabeling. Note that by Lemma~\ref{lem:indProb} we have
\begin{align}\label{eq:actualVal}
\bra{y_0}\sigma_0\ket{y_0} &\leq \bra{0}\sigma_0\ket{0} = \frac{1}{2} + \frac{1}{2\sqrt{d}}\ ,
\end{align}
for all $y_0$ and thus we choose $m = 0$ in~\eqref{eq:maxM}. Note that by construction we have
$\sigma_{y_0} = \xd^{y_0} \sigma_0 (\xd^{y_0})^\dagger$,
and hence
$\bra{y_0}\sigma_{y_0}\ket{y_0} = \bra{0}\sigma_0\ket{0}$.
Thus for the measurement in the computational basis given by
$M_{y_0} = \proj{y_0}$,
the inequality~\eqref{eq:upperBoundVal} is tight which together with~\eqref{eq:actualVal} gives our claim.
The case of retrieving $Y_1$ is exactly analogous, with the roles of $\xd$ and $\zd$ interchanged.
\end{proof}

\subsection{Min-entropy splitting}

We are ready to show that the min-entropy splitting inequality~\eqref{eq:splitting} is violated for the ccq-state given in~\eqref{eq:source}.

\begin{theorem}
For the ccq-state given in~\eqref{eq:source}, we have that for any ccqc state $\rho_{Y_0Y_1EC}$ 
with $\dim(C) = 2$ satisfying
$\tr_{C}(\rho_{Y_0 Y_1 EC}) = \rho_{Y_0 Y_1 E}$,
\begin{align}
P_{\rm guess}(Y_c|E C=c) \geq \frac{1}{2} + \frac{1}{2 \sqrt{d}}\ 
\end{align}
for all $c \in \01$.
\end{theorem}
\begin{proof}
Note that we may express
\begin{align}
\rho_{Y_0Y_1 E C} = \frac{1}{d^2} \sum_{y_0 y_1} \proj{y_0} \otimes \proj{y_1} \otimes \rho_{y_0y_1c}^{EC}\ .
\end{align}
We now first note that by the reduced trace condition 
and the fact that $Y_0$ and $Y_1$ are classical
we must have that $tr_{C}(\rho_{y_0y_1c}^{EC}) = \senc$. Since $\senc$ is a pure state, this implies
that $\rho_{y_0y_1c}^{EC} = \senc \otimes \sigma_{y_0 y_1}^C$. Since $C$ is classical
which we can express $\sigma_{y_0y_1}^C$ without loss of generality in the computational basis as
\begin{align}
\sigma_{y_0y_1}^C = q_{y_0y_1} \proj{0} + (1 - q_{y_0y_1}) \proj{1}\ ,
\end{align}
for some arbitrary distribution $\{q_{y_0y_1},1-q_{y_0y_1}\}$. 

Let us now consider how well we can compute $P_{\rm guess}(Y_0|E C = 0)$; the case of $C = 1$ is analogous.
First of all, note that the state obtained from $\rho_{Y_0Y_1EC}$ 
after we measured $C$ in the computational basis and obtained outcome $C=0$, followed by tracing out $C$ is given by
\begin{align}
	&\rho_{Y_0Y_1E}\nonumber\\
	&= \frac{1}{\hat{q}_0} \sum_{y_0,y_1} \tilde{q}_{y_0} \tilde{q}_{y_1|y_0} \proj{y_0}\otimes \proj{y_1} \otimes \senc\ ,
\end{align}
where $\hat{q}_0 = \sum_{y_0y_1} \tilde{q}_{y_0y_1}$ and $\tilde{q}_{y_0y_1} = (1/d^2) q_{y_0} q_{y_1|y_0}$.
The states we wish to distinguish are thus given by
\begin{align}
	\sigma_{y_0|c=0} := \frac{1}{\sum_{y_1} \tilde{q}_{y_1|y_0}} \sum_{y_1} \tilde{q}_{y_1|y_0} \rho_{y_0y_1}\ ,
\end{align}
Note that from Lemma~\ref{lem:indProb} we have that for all $y_0$
\begin{align}
\bra{y_0}\sigma_{y_0|c=0}\ket{y_0} = \frac{1}{2} + \frac{1}{2 \sqrt{d}}\ . 
\end{align}
Hence, for the measurement in the computational basis we succeed with probability at least $1/2 + 1/(2 \sqrt{d})$, independent of the
distributions $\{q_{y_0y_1},1-q_{y_0y_1}\}$.
Again, by exchanging the roles of $\xd$ and $\zd$ the same probability can be achieved using a measurement
in the Fourier basis, which proves the theorem.
\end{proof}

In terms of min-entropy, we thus have that $\hmin(Y_0 Y_1|E) = \log d$ but for all $C$ we have $\hmin(Y_C|EC) \approx 1$!
This effect is still observed for the $\eps$-smooth min-entropy for small $\eps$, since 
\begin{align}
\hmineps(Y_0 Y_1|E) \geq \hmin(Y_0 Y_1|E) = \log d\ ,
\end{align}
and
\begin{align}
- \log\left(P_{\rm guess}(Y_C|EC) - \eps\right) \geq \hmineps(Y_C|EC)\ .
\end{align}


\begin{thebibliography}{21}
\expandafter\ifx\csname natexlab\endcsname\relax\def\natexlab#1{#1}\fi
\expandafter\ifx\csname bibnamefont\endcsname\relax
  \def\bibnamefont#1{#1}\fi
\expandafter\ifx\csname bibfnamefont\endcsname\relax
  \def\bibfnamefont#1{#1}\fi
\expandafter\ifx\csname citenamefont\endcsname\relax
  \def\citenamefont#1{#1}\fi
\expandafter\ifx\csname url\endcsname\relax
  \def\url#1{\texttt{#1}}\fi
\expandafter\ifx\csname urlprefix\endcsname\relax\def\urlprefix{URL }\fi
\providecommand{\bibinfo}[2]{#2}
\providecommand{\eprint}[2][]{\url{#2}}

\bibitem[{\citenamefont{Spekkens}(2005)}]{Spekkens2005}
\bibinfo{author}{\bibfnamefont{R.}~\bibnamefont{Spekkens}},
  \bibinfo{journal}{Physical Review A} \textbf{\bibinfo{volume}{71}},
  \bibinfo{pages}{052108} (\bibinfo{year}{2005}).

\bibitem[{\citenamefont{K\"onig et~al.}(2009)\citenamefont{K\"onig, Renner, and
  Schaffner}}]{krs:entropy}
\bibinfo{author}{\bibfnamefont{R.}~\bibnamefont{K\"onig}},
  \bibinfo{author}{\bibfnamefont{R.}~\bibnamefont{Renner}}, \bibnamefont{and}
  \bibinfo{author}{\bibfnamefont{C.}~\bibnamefont{Schaffner}},
  \bibinfo{journal}{IEEE Trans. Info.} \textbf{\bibinfo{volume}{55}}
  (\bibinfo{year}{2009}).

\bibitem[{\citenamefont{Wullschleger}(2007)}]{juerg:splitting}
\bibinfo{author}{\bibfnamefont{J.}~\bibnamefont{Wullschleger}}, in
  \emph{\bibinfo{booktitle}{Advances in Cryptology --- {EUROCRYPT}~'07}}
  (\bibinfo{publisher}{Springer-Verlag}, \bibinfo{year}{2007}), Lecture Notes
  in Computer Science.

\bibitem[{\citenamefont{Damg{\aa}rd et~al.}(2007)\citenamefont{Damg{\aa}rd,
  Fehr, Renner, Salvail, and Schaffner}}]{serge:new}
\bibinfo{author}{\bibfnamefont{I.~B.} \bibnamefont{Damg{\aa}rd}},
  \bibinfo{author}{\bibfnamefont{S.}~\bibnamefont{Fehr}},
  \bibinfo{author}{\bibfnamefont{R.}~\bibnamefont{Renner}},
  \bibinfo{author}{\bibfnamefont{L.}~\bibnamefont{Salvail}}, \bibnamefont{and}
  \bibinfo{author}{\bibfnamefont{C.}~\bibnamefont{Schaffner}}, in
  \emph{\bibinfo{booktitle}{Advances in Cryptology---CRYPTO~'07}}
  (\bibinfo{publisher}{Springer-Verlag}, \bibinfo{year}{2007}), vol.
  \bibinfo{volume}{4622} of \emph{\bibinfo{series}{Lecture Notes in Computer
  Science}}, pp. \bibinfo{pages}{360--378}.

\bibitem[{\citenamefont{Nayak}(1999)}]{nayak:rac}
\bibinfo{author}{\bibfnamefont{A.}~\bibnamefont{Nayak}}, in
  \emph{\bibinfo{booktitle}{Proceedings of 40th IEEE FOCS}}
  (\bibinfo{year}{1999}), pp. \bibinfo{pages}{369--376}.

\bibitem[{\citenamefont{Ambainis et~al.}(1999)\citenamefont{Ambainis, Nayak,
  {Ta-Shma}, and Vazirani}}]{nayak:original}
\bibinfo{author}{\bibfnamefont{A.}~\bibnamefont{Ambainis}},
  \bibinfo{author}{\bibfnamefont{A.}~\bibnamefont{Nayak}},
  \bibinfo{author}{\bibfnamefont{A.}~\bibnamefont{{Ta-Shma}}},
  \bibnamefont{and} \bibinfo{author}{\bibfnamefont{U.}~\bibnamefont{Vazirani}},
  in \emph{\bibinfo{booktitle}{Proceedings of 31st ACM STOC}}
  (\bibinfo{year}{1999}), pp. \bibinfo{pages}{376--383}.

\bibitem[{\citenamefont{Klyachko et~al.}(2008)\citenamefont{Klyachko, Can,
  Binicio\u{g}lu, and Shumsovsky}}]{klyachko:nc}
\bibinfo{author}{\bibfnamefont{A.~A.} \bibnamefont{Klyachko}},
  \bibinfo{author}{\bibfnamefont{M.~A.} \bibnamefont{Can}},
  \bibinfo{author}{\bibfnamefont{S.}~\bibnamefont{Binicio\u{g}lu}},
  \bibnamefont{and} \bibinfo{author}{\bibfnamefont{A.~S.}
  \bibnamefont{Shumsovsky}}, \bibinfo{journal}{Physical Review Letters}
  \textbf{\bibinfo{volume}{101}}, \bibinfo{pages}{020403}
  (\bibinfo{year}{2008}).

\bibitem[{\citenamefont{Cabello et~al.}(2010)\citenamefont{Cabello, Severini,
  and Winter}}]{andreas:nc}
\bibinfo{author}{\bibfnamefont{A.}~\bibnamefont{Cabello}},
  \bibinfo{author}{\bibfnamefont{S.}~\bibnamefont{Severini}}, \bibnamefont{and}
  \bibinfo{author}{\bibfnamefont{A.}~\bibnamefont{Winter}}
  (\bibinfo{year}{2010}), \bibinfo{note}{arXiv:1010.2163}.

\bibitem[{\citenamefont{Clauser et~al.}(1969)\citenamefont{Clauser, Horne,
  Shimony, and Holt}}]{chsh}
\bibinfo{author}{\bibfnamefont{J.}~\bibnamefont{Clauser}},
  \bibinfo{author}{\bibfnamefont{M.}~\bibnamefont{Horne}},
  \bibinfo{author}{\bibfnamefont{A.}~\bibnamefont{Shimony}}, \bibnamefont{and}
  \bibinfo{author}{\bibfnamefont{R.}~\bibnamefont{Holt}},
  \bibinfo{journal}{Phys. Rev. Lett.} \textbf{\bibinfo{volume}{23}},
  \bibinfo{pages}{880} (\bibinfo{year}{1969}).

\bibitem[{\citenamefont{Spekkens et~al.}(2009)\citenamefont{Spekkens, Buzacott,
  Keehn, Toner, and Pryde}}]{Spekkens2009}
\bibinfo{author}{\bibfnamefont{R.}~\bibnamefont{Spekkens}},
  \bibinfo{author}{\bibfnamefont{D.}~\bibnamefont{Buzacott}},
  \bibinfo{author}{\bibfnamefont{A.}~\bibnamefont{Keehn}},
  \bibinfo{author}{\bibfnamefont{B.}~\bibnamefont{Toner}}, \bibnamefont{and}
  \bibinfo{author}{\bibfnamefont{G.}~\bibnamefont{Pryde}},
  \bibinfo{journal}{Physical Review Letters} \textbf{\bibinfo{volume}{102}},
  \bibinfo{pages}{010401} (\bibinfo{year}{2009}).

\bibitem[{\citenamefont{Schaffner}(2007)}]{chris:talk}
\bibinfo{author}{\bibfnamefont{C.}~\bibnamefont{Schaffner}}
  (\bibinfo{year}{2007}), \bibinfo{note}{{P}ersonal communication}.

\bibitem[{\citenamefont{Damg{\aa}rd et~al.}(2005)\citenamefont{Damg{\aa}rd,
  Fehr, Salvail, and Schaffner}}]{serge:bounded}
\bibinfo{author}{\bibfnamefont{I.~B.} \bibnamefont{Damg{\aa}rd}},
  \bibinfo{author}{\bibfnamefont{S.}~\bibnamefont{Fehr}},
  \bibinfo{author}{\bibfnamefont{L.}~\bibnamefont{Salvail}}, \bibnamefont{and}
  \bibinfo{author}{\bibfnamefont{C.}~\bibnamefont{Schaffner}}, in
  \emph{\bibinfo{booktitle}{Proceedings of 46th IEEE FOCS}}
  (\bibinfo{year}{2005}), pp. \bibinfo{pages}{449--458}.

\bibitem[{\citenamefont{K{\"o}nig and Renner}(2007)}]{kr:sampling}
\bibinfo{author}{\bibfnamefont{R.}~\bibnamefont{K{\"o}nig}} \bibnamefont{and}
  \bibinfo{author}{\bibfnamefont{R.}~\bibnamefont{Renner}}
  (\bibinfo{year}{2007}), \bibinfo{note}{arXiv:0712.4291}.

\bibitem[{\citenamefont{K{\"o}nig et~al.}(2009)\citenamefont{K{\"o}nig, Wehner,
  and Wullschleger}}]{KoeWehWul09}
\bibinfo{author}{\bibfnamefont{R.}~\bibnamefont{K{\"o}nig}},
  \bibinfo{author}{\bibfnamefont{S.}~\bibnamefont{Wehner}}, \bibnamefont{and}
  \bibinfo{author}{\bibfnamefont{J.}~\bibnamefont{Wullschleger}}
  (\bibinfo{year}{2009}), \bibinfo{note}{arXiv:0906.1030}.

\bibitem[{\citenamefont{Schr{\"o}dinger}(1935)}]{schroedinger:eprGerman}
\bibinfo{author}{\bibfnamefont{E.}~\bibnamefont{Schr{\"o}dinger}},
  \bibinfo{journal}{Naturwissenschaften} \textbf{\bibinfo{volume}{23}},
  \bibinfo{pages}{807,823,840} (\bibinfo{year}{1935}).

\bibitem[{\citenamefont{Oppenheim and Wehner}(2010{\natexlab{a}})}]{js:urvsnl}
\bibinfo{author}{\bibfnamefont{J.}~\bibnamefont{Oppenheim}} \bibnamefont{and}
  \bibinfo{author}{\bibfnamefont{S.}~\bibnamefont{Wehner}},
  \bibinfo{journal}{Science} \textbf{\bibinfo{volume}{330}},
  \bibinfo{pages}{1072} (\bibinfo{year}{2010}{\natexlab{a}}).

\bibitem[{\citenamefont{Oppenheim and Wehner}(2010{\natexlab{b}})}]{js:inprep}
\bibinfo{author}{\bibfnamefont{J.}~\bibnamefont{Oppenheim}} \bibnamefont{and}
  \bibinfo{author}{\bibfnamefont{S.}~\bibnamefont{Wehner}}
  (\bibinfo{year}{2010}{\natexlab{b}}), \bibinfo{note}{in preparation}.

\bibitem[{\citenamefont{Short and Wehner}(2010)}]{ts:entropy}
\bibinfo{author}{\bibfnamefont{A.~J.} \bibnamefont{Short}} \bibnamefont{and}
  \bibinfo{author}{\bibfnamefont{S.}~\bibnamefont{Wehner}},
  \bibinfo{journal}{New Journal of Physics} \textbf{\bibinfo{volume}{12}},
  \bibinfo{pages}{033023} (\bibinfo{year}{2010}).

\bibitem[{\citenamefont{Renner}(2005)}]{renato:diss}
\bibinfo{author}{\bibfnamefont{R.}~\bibnamefont{Renner}}, Ph.D. thesis,
  \bibinfo{school}{ETH Zurich} (\bibinfo{year}{2005}),
  \bibinfo{note}{quant-ph/0512258}.

\bibitem[{\citenamefont{Bowie}(2003)}]{philCircle}
\bibinfo{author}{\bibfnamefont{A.}~\bibnamefont{Bowie}},
  \emph{\bibinfo{title}{Introduction to German philosophy: from Kant to
  Habermas}} (\bibinfo{year}{2003}).

\bibitem[{\citenamefont{Odendaal and Plastino}(2010)}]{odendaal}
\bibinfo{author}{\bibfnamefont{R.~Q.} \bibnamefont{Odendaal}} \bibnamefont{and}
  \bibinfo{author}{\bibfnamefont{A.~R.} \bibnamefont{Plastino}},
  \bibinfo{journal}{Eur. J. of Physics} \textbf{\bibinfo{volume}{31}},
  \bibinfo{pages}{193} (\bibinfo{year}{2010}).

\bibitem[]{endnote23}{Intuitively, a non-contextual model is one in which the observable statistics of a particular measurement do not depend on the \protect \emph  {context} in which the measurement is performed, and in particular on which other compatible measurements are possibly performed simultaneously. Here, and as is usual, we consider non-contextual models in which the measurements are composed of deterministic effects. We refer to the appendix for formal definitions.}
\bibitem[]{endnoteAdd}{Recall that there are infinitely many primes.}
\bibitem[]{endnote25}{In principle, $C$ could be arbitrary, but since we are only interested in the result of a 2-outcome measurement on $C$, we assume that $C$ is indeed already classical and use the subscript $C$ to denote that classical value.}
\bibitem[]{endnote26}{The whole here being a maximally entangled state, and the parts being the individual (locally completely mixed) subsystems. See also\protect \nobreakspace  {}\cite {odendaal} for an example.}
\bibitem[]{endnote27}{For example, one could consider weak forms of oblivious transfer where one only demands security against the receiver.}
\bibitem[]{endnote30}{A simple example in the classical setting --- storing each of the halves with probability half --- should convince the reader that this is necessary.}
\bibitem[]{endnote28}{We will write a $\protect \qopname  \relax m{max}$ instead of a $\protect \qopname  \relax m{sup}$ everywhere, and assume that the dimension of the system $E$ is finite.}
\bibitem[]{endnote29}{Some of these variables may be hidden, in the sense that the operational description of a particular preparation does not necessarily fully determine the distribution on variables that describe it, but we will not need to make that distinction.}

\end{thebibliography}
\end{document}